\documentclass[11pt]{amsart}
\usepackage{amsthm}
\usepackage{amsmath,amstext,amsthm,amsfonts,amssymb,amsthm}
\usepackage{multicol}
\usepackage{latexsym}
\usepackage{color}
\usepackage{graphicx}
\usepackage{epsf}
\usepackage{epsfig}
\usepackage{epic}
\usepackage{graphics}
\usepackage{verbatim}
\usepackage{color}
\usepackage{stmaryrd}
\usepackage{hyperref}
\newtheorem{theorem}{Theorem}[section]
\newtheorem{lemma}[theorem]{Lemma}
\theoremstyle{definition}
\newtheorem{definition}[theorem]{Definition}
\newtheorem{proposition}[theorem]{Proposition}

\theoremstyle{remark}
\newtheorem{remark}[theorem]{Remark}
\numberwithin{equation}{section}

\newcommand{\quat}{\mathbb H}

\newcommand{\an}{\mathfrak{a}}

\newcommand{\q}{\mathbf{q}}
\newcommand{\I}{\mathbb{I}}
\newcommand{\HI}{\mathfrak{H}}

\newcommand{\D}{\mathfrak{D}}
\newcommand{\C}{\mathbb{C}}

\newcommand{\qu}{\mathfrak{q}}
\newcommand{\pu}{\mathfrak{p}}
\newcommand{\as}{\mathsf{a}}
\newcommand{\asd}{\mathsf{a}^{\dagger}}

\newcommand{\oqu}{\overline{\mathfrak{q}}}
\newcommand{\opu}{\overline{\mathfrak{p}}}
\newcommand{\HIB}{\mathfrak{H}^B_r}

\begin{document}
\title[Squeezed states ]{
	Squeezed states in the quaternionic setting}
\author{ K. Thirulogasanthar$^{\dagger}$, B. Muraleetharan$^{\ddagger}$}
\address{$^{\ddagger}$ Department of Computer Science and Software Engineering, Concordia University, 1455 De Maisonneuve Blvd. West, Montreal, Quebec, H3G 1M8, Canada.}
\address{$^{\dagger}$ Department of mathematics and Statistics, University of Jaffna, Thirunelveli, Sri Lanka.}
\email{santhar@gmail.com and bbmuraleetharan@jfn.ac.lk}
\date{\today}
\thanks{K. Thirulogasanthar would like to thank the, FQRNT, Fonds de la Recherche  Nature et  Technologies (Quebec, Canada) for partial financial support Under the grant number 2017-CO-201915. Part of this work was done while he was visiting the Politecnico di Milano to which he expresses his thanks for the hospitality. He also thanks the program Professori Visitatori GNSAGA, INDAM for the support during the period in which this paper was partially written. The authors thank Prof. I Sabadini for discussions.}
%%%%%%%%%%%%%%%%%%%%%%%%%%%%%%%%%%%%%%%%%%%%%%%%%%%%%%%%%%%%%%%%%%%%%%%%%%%%%%%%%%%%%%%%%%%%%%%%%%%%%%%%%%%%%
\pagestyle{myheadings}
%%%%%%%%%%%%%%%%%%%%%%%%%%%%%%%%%%%%%%%%%%%%%%%%%%%%%%%%%%%%%
%\begin{multicols}{2}
\begin{abstract}
  Using a left multiplication defined on a right quaternionic Hilbert space, we shall demonstrate that pure squeezed states can be defined with all the desired properties on a right quaternionic Hilbert space. Further, we shall also  demonstrate squeezed states can be defined on the same Hilbert space, but the noncommutativity of quaternions prevent us in getting the desired results.  However, we will show that if once considers the quaternionic slice wise approach, then the desired properties can be obtained for quaternionic squeezed states.
\end{abstract}
\maketitle
%%%%%%%%%%%%%%%%%%%%%%%%%%%%%%%%%%%%%%%%%%%%%%%%%%%%%%%%%%%%%%%%%%%%%%%%%%%%%%%%%%%%%%%%%%%%%%%%%%%%%%%%%%%%%
\section{Introduction}\label{sec_intro}
As it is well known, quantum mechanics can be formulated over the complex and the quaternionic numbers, see \cite{bvn}. In recent times, new mathematical tools in quaternionic analysis became available in the literature and, as a consequence, there has been a resurgence of interest for the quaternionic quantum mechanics.
In this formulation, in complete analogy with the complex formulation,  states  are represented by vectors of a separable quaternionic Hilbert space and observables are represented by quaternionic linear and self-adjoint operators, see e.g. the celebrated book \cite{Ad} for more information.

However, until the most recent times, an appropriate spectral theory was missing since there was not a satisfactory notion of spectrum. This difficulty has been solved with the introduction of the notion of S-spectrum (see \cite{NFC}) and, accordingly, with a proof of the spectral theorem for normal operators, see \cite{ack}.

In our paper \cite{MTI}, we offered a discussion of the various notions of momentum operator and we show that, by using the notion of left multiplication in a right quaternionic Hilbert space it is possible to define a linear self-adjoint momentum operator in complete analogy with the complex case. The possibility to introduce a left multiplication in a right quaternionic Hilbert space is very well known and a very useful tool in several cases. In fact a linear space over the quaternions is, in general, on one side (either left or right). However, in order to have good properties when considering linear operators acting on the space, it is necessary to have a multiplication on both sides. It can be always defined but it requires to fix a basis, thus all the consequences that one may deduce, have to be shown to be independent of the choice of the basis. However, when we consider a particular quantum system we always work with a fixed basis, which is the wavefunctions of the Hamiltonian (Fock space basis). Therefore we do not need to worry about working with a fixed basis.

In \cite{MTI} we have also deepened the study of an appropriate harmonic oscillator displacement operator  showing that this displacement operator leads to a square integrable, irreducible and unitary representation and that it satisfies most of the properties of its complex counterpart.

In this paper we introduce and study the squeeze operator which is formally defined as in the complex setting but where the operation involved in the definition have to be interpreted in an appropriate way. To be specific,  a squeeze operator is obtained by exponentiating $\frac{1}{2}(\pu\cdot (\asd)^2-\overline{\pu}\cdot\as^2)$  where
$\pu\cdot $ is the left multiplication by the quaternion $\pu$  and $\asd, \as$ are the creation and annihilation operator. We show that this latter operator is anti-hermitian and we study several properties of the squeezed operator. We also study quaternionic pure squeezed states, obtained by the action of the squeeze operator on the vacuum state.

Due to the non-commutative nature of quaternions, there is an intrinsic issue if one is aimed to obtain relations involving both the displacement and the squeezed operator. Suitable relations can be obtained only slice-wise.

There is a vast interest in squeezed states in various applications, particularly in the coding and transmission of information through optical devices \cite{Gaz, Lo, Yuen}. In the quaternion case, these squeezed states appear as two component states in four variables. Hence these states have move degrees of freedom and may be useful in application.

The plan of the paper is as follows. Section 2 contains some preliminaries on quaternions, right quaternionic Hilbert spaces and the notion  of left multiplication. Section 3 studies the Bargmann space of regular functions, the displacement operator, the squeeze operator and some of its properties. We also introduce some quaternionic Lie algebras constructed by taking some suitable real or complex linear spaces and equipping them with suitable Lie brackets. The expectation values and the variances of the creation and annihilation operator and of the quadrature operators are computed in this section. The fourth section is devoted to the relation involving displacement, squeeze, creation an annihilation  operators. We obtain a result similar to the one in the complex case, but due the noncommutativity  it can be proved just on quaternionic slices.

%%%%%%%%%%%%%%%%%%%%%%%%%%%%%%%%%%%%%%%%%%%%%%%%%%%%%%%%%%%%%%%%%%%%%%%%%%%%%%%%%%%%%%%%%%%%%%%%%%%%%%%%%%%%%%
\section{Mathematical preliminaries}
In this section we recall some basic facts about quaternions, their complex matrix representation, quaternionic Hilbert spaces as needed here. For details we refer the reader to \cite{Ad, Vis, Za}.
\subsection{Quaternions}
Let $\quat$ denote the field of quaternions. Its elements are of the form $\qu=q_0+q_1i+q_2j+q_3k$ where $q_0,q_1,q_2$ and $q_3$ are real numbers, and $i,j,k$ are imaginary units such that $i^2=j^2=k^2=-1$, $ij=-ji=k$, $jk=-kj=i$ and $ki=-ik=j$. The quaternionic conjugate of $\qu$ is defined to be $\overline{\qu} = q_0 - q_1i - q_2j - q_3k$. Quaternions can be represented by $2\times 2$ complex matrices:
\begin{equation}
\qu = q_0 \sigma_{0} + i \q \cdot \underline{\sigma},
\end{equation}
with $q_0 \in \mathbb R , \quad \q = (q_1, q_2, q_3)
\in \mathbb R^3$, $\sigma_0 = \mathbb{I}_2$, the $2\times 2$ identity matrix, and
$\underline{\sigma} = (\sigma_1, -\sigma_2, \sigma_3)$, where the
$\sigma_\ell, \; \ell =1,2,3$ are the usual Pauli matrices. The quaternionic imaginary units are identified as, $i = \sqrt{-1}\sigma_1, \;\; j = -\sqrt{-1}\sigma_2, \;\; k = \sqrt{-1}\sigma_3$. Thus,
\begin{equation}
\qu = \left(\begin{array}{cc}
q_0 + i q_3 & -q_2 + i q_1 \\
q_2 + i q_1 & q_0 - i q_3
\end{array}\right) \qquad %\text{and} \qquad \overline{\qu} = \qu^\dag\quad \text{(matrix adjoint)}\; .
\label{q3}
\end{equation}
and $\overline{\qu} = \qu^\dag\quad \text{(matrix adjoint)}\; .$
Using the polar coordinates:
\begin{eqnarray*}
	q_0 &=& r \cos{\theta}, \\
	q_1 &=& r \sin{\theta} \sin{\phi} \cos{\psi}, \\
	q_2 &=& r \sin{\theta} \sin{\phi} \sin{\psi}, \\
	q_3 &=& r \sin{\theta} \cos{\phi},
\end{eqnarray*}
where $(r,\phi,\theta,\psi)  \in [0,\infty)\times[0,\pi]\times[0,2\pi)^{2}$, we may write
\begin{equation}
\qu = A(r) e^{i \theta \sigma(\widehat{n})}
\label{q4},
\end{equation}
where
\begin{equation}
A(r) = r\mathbb \sigma_0
\end{equation}
and\begin{equation}
\sigma(\widehat{n}) = \left(\begin{array}{cc}
\cos{\phi} & \sin{\phi} e^{i\psi} \\
\sin{\phi} e^{-i\psi} & -\cos{\phi}
\end{array}\right).
\label{q5}\end{equation}
The matrices
$A(r)$ and $\sigma(\widehat{n})$ satisfy the conditions,
\begin{equation}
A(r) = A(r)^\dagger,~\sigma(\widehat{n})^2 = \sigma_0
,~\sigma(\widehat{n})^\dagger = \sigma(\widehat{n})
\label{san1}
\end{equation}
and
$\lbrack A(r), \sigma(\widehat{n}) \rbrack = 0.$
Note that a real norm on $\quat$ is defined by  $$\vert\qu\vert^2  := \overline{\qu} \qu = r^2 \sigma_0 = (q_0^2 +  q_1^2 +  q_2^2 +  q_3^2).$$
%A typical measure on $\quat$ may take the form
%\begin{equation}\label{measure}
%d\varsigma(r, \theta,\phi, \psi)= d\tau(r)\, d\theta\, d\Omega(\phi ,\psi )
%\end{equation}
%with $d\Omega(\phi ,\psi) = \displaystyle{\frac{1}{4\pi}} \,\sin{\phi}\, d\phi \,d\psi .$
Note also that for ${\pu},\qu\in \quat$, we have $\overline{{\pu}\qu}=\overline{\qu}~\overline{{\pu}}$, $\pu\qu\not=\qu\pu$, $\qu\overline{{ \qu}}=\overline{{\qu}}\qu$, and real numbers commute with quaternions.
Quaternions can also be interpreted as a sum of scalar and a vector by writing $$\qu=q_0+q_1i+q_2j+q_3k=(q_0,\q);$$ where $\q=q_1i+q_2j+q_3k$. This particular method of writing quaternions will help us for forming a representation of Weyl-Heisenberg Lie algebra under the quaternionic settings.
We borrow the materials as needed here from \cite{Gen1}.  Let
\begin{eqnarray*}
	\mathbb{S}&=&\{I=x_1i+x_2j+x_3k~\vert
	~x_1,x_2,x_3\in\mathbb{R},~x_1^2+x_2^2+x_3^2=1\},
\end{eqnarray*}
we call it a quaternion sphere.
 \begin{proposition}\cite{Gen1}\label{Pr1}
	For any non-real quaternion $\qu\in \quat\smallsetminus\mathbb{R}$, there exist, and are unique, $x,y\in\mathbb{R}$ with $y>0$, and $I_\qu\in\mathbb{S}$ such that $\qu=x+I_\qu y$.
\end{proposition}
For every quaternion $I\in\mathbb{S}$, the complex line $\C_I=\mathbb{R}+I\mathbb{R}$ passing through the origin, and containing $1$ and $I$, is called a quaternion slice. Thereby, we can see that
\begin{equation}\label{Eq1}
\quat=\bigcup_{I\in\mathbb{S}}\C_I\quad\text{and}\quad\bigcap_{I\in\mathbb{S}} \C_I=\mathbb{R}
\end{equation}
One can also easily see that $\C_I\subset \quat$ is commutative, while, elements from two different quaternion slices, $\C_I$ and $\C_J$ (for $I, J\in\mathbb{S}$ with $I\not=J$), do not necessarily commute.
%%%%%%%%%%%%%%%%%%%%%%%%%%%%%%%%%%%%%%%%%%%%%%%%%%%%%%
\subsection{Quaternionic Hilbert spaces}
In this subsection we introduce right quaternionic Hilbert spaces. For details we refer the reader to \cite{Ad}. We also define the Hilbert space of square integrable functions on quaternions based on \cite{Vis, Gu}.
\subsubsection{Right Quaternionic Hilbert Space}
Let $V_{\quat}^{R}$ be a linear vector space under right multiplication by quaternionic scalars (again $\quat$ standing for the field of quaternions).  For $f,g,h\in V_{\quat}^{R}$ and $\qu\in \quat$, the inner product
$$\langle\cdot\mid\cdot\rangle:V_{\quat}^{R}\times V_{\quat}^{R}\longrightarrow \quat$$
satisfies the following properties
\begin{enumerate}
	\item[(i)]
	$\overline{\langle f\mid g\rangle}=\langle g\mid f\rangle$
	\item[(ii)]
	$\|f\|^{2}=\langle f\mid f\rangle>0$ unless $f=0$, a real norm
	\item[(iii)]
	$\langle f\mid g+h\rangle=\langle f\mid g\rangle+\langle f\mid h\rangle$
	\item[(iv)]
	$\langle f\mid g\qu\rangle=\langle f\mid g\rangle\qu$
	\item[(v)]
	$\langle f\qu\mid g\rangle=\overline{\qu}\langle f\mid g\rangle$
\end{enumerate}
where $\overline{\qu}$ stands for the quaternionic conjugate. We assume that the
space $V_{\quat}^{R}$ is complete under the norm given above. Then,  together with $\langle\cdot\mid\cdot\rangle$ this defines a right quaternionic Hilbert space, which we shall assume to be separable. Quaternionic Hilbert spaces share most of the standard properties of complex Hilbert spaces. In particular, the Cauchy-Schwartz inequality holds on quaternionic Hilbert spaces as well as the Riesz representation theorem for their duals.  Thus, the Dirac bra-ket notation
can be adapted to quaternionic Hilbert spaces:
$$\mid f\qu\rangle=\mid f\rangle\qu,\hspace{1cm}\langle f\qu\mid=\overline{\qu}\langle f\mid\;, $$
for a right quaternionic Hilbert space, with $\vert f\rangle$ denoting the vector $f$ and $\langle f\vert$ its dual vector. Similarly the left quaternionic Hilbert space $V_{\quat}^{L}$ can also be described, see for more detail \cite{Ad,MuTh,Thi1}.
The field of quaternions $\quat$ itself can be turned into a left quaternionic Hilbert space by defining the inner product $\langle \qu \mid \qu^\prime \rangle = \qu \qu^{\prime\dag} = \qu\overline{\qu^\prime}$ or into a right quaternionic Hilbert space with  $\langle \qu \mid \qu^\prime \rangle = \qu^\dag \qu^\prime = \overline{\qu}\qu^\prime$. Further note that, due to the non-commutativity of quaternions the sum
$\sum_{m=0}^{\infty}\pu^m\qu^m/m!$
cannot be written as $\text{exp}(\pu\qu).$ However, in any Hilbert space the norm convergence implies the convergence of the series and
$\sum_{m=0}^{\infty}\left\vert \pu^m\qu^m/m!\right\vert
\leq e^{|\pu||\qu|},$ therefore $\sum_{m=0}^{\infty}\pu^m\qu^m/m!=e^{\pu\qu}_*$ converges.
%%%%%%%%%%%%%%%%%%%%%%%%%%%%%%%%%%%%%%%%%%%%%%%%%%%%%%
\subsubsection{Quaternionic Hilbert Spaces of Square Integrable Functions}
Let $(X, \mu)$ be a measure space and $\quat$  the field of quaternions, then
$$L^2_\quat(X, d\mu)=\left\{f:X\rightarrow \quat \left|  \int_X|f(x)|^2d\mu(x)<\infty \right.\right\}$$
\noindent
is a right quaternionic Hilbert space which is denoted by $L^2_\quat(X,\mu)$, with the (right) scalar product
\begin{equation}
\langle f \mid g\rangle =\int_X\overline{ f(x)}{g(x)} d\mu(x),
\label{left-sc-prod}
\end{equation}
where $\overline{f(x)}$ is the quaternionic conjugate of $f(x)$, and (right)  scalar multiplication $f\an, \; \an\in \quat,$ with $(f\an)(\qu) = f(\qu)\an$ (see \cite{Gu,Vis} for details). Similarly, one could define a left quaternionic Hilbert space of square integrable functions.
\subsection{Left Scalar Multiplications on $V_{\quat}^{R}$.}
We shall extract the definition and some properties of left scalar multiples of vectors on $V_{\quat}^R$ from \cite{ghimorper} as needed for the development of the manuscript. The left scalar multiple of vectors on a right quaternionic Hilbert space is an extremely non-canonical operation associated with a choice of preferred Hilbert basis. Now the Hilbert space $V_{\quat}^{R}$ has a Hilbert basis
\begin{equation}\label{b1}
\mathcal{O}=\{\varphi_{k}\,\mid\,k\in N\},
\end{equation}
where $N$ is a countable index set.
The left scalar multiplication `$\cdot$' on $V_{\quat}^{R}$ induced by $\mathcal{O}$ is defined as the map $\quat\times V_{\quat}^{R}\ni(\qu,\phi)\longmapsto \qu\cdot\phi\in V_{\quat}^{R}$ given by
\begin{equation}\label{LPro}
\qu\cdot\phi:=\sum_{k\in N}\varphi_{k}\qu\langle \varphi_{k}\mid \phi\rangle,
\end{equation}
for all $(\qu,\phi)\in\quat\times V_{\quat}^{R}$. Since all left multiplications are made with respect to some basis, assume that the basis $\mathcal{O}$ given by (\ref{b1}) is fixed all over the paper.
\begin{proposition}\cite{ghimorper}\label{lft_mul}
	The left product defined in (\ref{LPro}) satisfies the following properties. For every $\phi,\psi\in V_{\quat}^{R}$ and $\pu,\qu\in\quat$,
	\begin{itemize}
		\item[(a)] $\qu\cdot(\phi+\psi)=\qu\cdot\phi+\qu\cdot\psi$ and $\qu\cdot(\phi\pu)=(\qu\cdot\phi)\pu$.
		\item[(b)] $\|\qu\cdot\phi\|=|\qu|\|\phi\|$.
		\item[(c)] $\qu\cdot(\pu\cdot\phi)=(\qu\pu\cdot\phi)$.
		\item[(d)] $\langle\overline{\qu}\cdot\phi\mid\psi\rangle
		=\langle\phi\mid\qu\cdot\psi\rangle$.
		\item[(e)] $r\cdot\phi=\phi r$, for all $r\in \mathbb{R}$.
		\item[(f)] $\qu\cdot\varphi_{k}=\varphi_{k}\qu$, for all $k\in N$.
	\end{itemize}
\end{proposition}
\begin{remark}%\label{Rem1}
 It is immediate that $(\pu+\qu)\cdot\phi=\pu\cdot\phi+\qu\cdot\phi$, for all $\pu,\qu\in\quat$ and $\phi\in V_{\quat}^{R}$. Moreover, with the aid of (b) in above Proposition (\ref{lft_mul}), we can have, if $\{\phi_n\}$ in $V_\quat^R$ such that $\phi_n\longrightarrow\phi$, then $\qu\cdot\phi_n\longrightarrow\qu\cdot\phi$. Also if $\sum_{n}\phi_n$ is a convergent sequence in $V_\quat^R$, then $\qu\cdot(\sum_{n}\phi_n)=\sum_{n}\qu\cdot\phi_n$.
\end{remark}
Furthermore, the quaternionic scalar multiplication of $\quat$-linear operators is also defined in \cite{ghimorper}. For any fixed $\qu\in\quat$ and a given right $\quat$-linear operator $A:\D(A)\longrightarrow V_{\quat}^{R}$, the left scalar multiplication `$\cdot$' of $A$ is defined as a map $\qu \cdot A:\D(A)\longrightarrow V_{\quat}^{R}$ by the setting
\begin{equation}\label{lft_mul-op}
(\qu\cdot A)\phi:=\qu\cdot (A\phi)=\sum_{k\in N}\varphi_{k}\qu\langle \varphi_{k}\mid A\phi\rangle,
\end{equation}
for all $\phi\in \D(A)$. It is straightforward that $\qu A$ is a right $\quat$-linear operator. If $\qu\cdot\phi\in \D(A)$, for all $\phi\in \D(A)$, one can define right scalar multiplication `$\cdot$' of the right $\quat$-linear operator $A:\D(A)\longrightarrow V_{\quat}^{R}$ as a map $ A\cdot\qu:\D(A)\longrightarrow V_{\quat}^{R}$ by the setting
\begin{equation}\label{rgt_mul-op}
(A\cdot\qu )\phi:=A(\qu\cdot \phi),
\end{equation}
for all $\phi\in \D(A)$. It is also right $\quat$-linear operator. One can easily obtain that, if $\qu\cdot\phi\in \D(A)$, for all $\phi\in \D(A)$ and $\D(A)$ is dense in $V_{\quat}^{R}$, then
\begin{equation}\label{sc_mul_aj-op}
(\qu\cdot A)^{\dagger}=A^{\dagger}\cdot\overline{\qu}~\mbox{~and~}~
(A\cdot\qu)^{\dagger}=\overline{\qu}\cdot A^{\dagger}.
\end{equation}
%%%%%%%%%%%%%%%%%%%%%%%%%%%%%%%%%%%%%%%%%%%%%%%%%%%%%%

\section{Bargmann space of regular and anti-regular functions}
The Bargmann space of left regular functions $\HI^B_r$ is a closed subspace of the right Hilbert space $L_{\quat}(\quat, d\zeta(r, \theta, \phi, \psi))$, where $d\zeta(r, \theta, \phi, \psi)=\frac{1}{4\pi} e^{-r^2}\sin{\phi} dr d\theta d\phi d\psi$. An orthonormal basis of this space is given by the monomials (which are both left and right regular)
$$\Phi_n(\qu)=\frac{\qu^n}{\sqrt{n!}};\quad n=0,1,2,\cdots.$$
Similarly the vectors
$$\overline{\Phi_n(\qu)}=\Phi_n(\oqu)=\frac{\oqu^n}{\sqrt{n!}};\quad n=0,1,2,\cdots$$
form an orthonormal basis in the corresponding space of right anti-regular functions $\HI^B_{ar}$. There is also an associated reproducing kernel
$$K_B(\qu,\overline{\pu})=\sum_{n=0}^{\infty}\Phi_n(\qu)\overline{\Phi_n(\pu)}=e_{\star}^{\qu\overline{\pu}}$$
 see \cite{Thi1, Al} for details.

\subsection{Coherent states on right quaternionic Hilbert spaces}\label{CSLQH}
The main content of this section is extracted from \cite{Thi2} as needed here. For an enhanced explanation we refer the reader to \cite{Thi2}. In \cite{Thi2} the authors have defined coherent states on $V_{\quat}^{R}$ and $V_{\quat}^{L}$, and also established the normalization and resolution of the identities for each of them.\\

On the Bargmann space $\HI^B_r$, the normalized canonical coherent states are
\begin{equation}\label{CCS}
\eta_{\qu}=\frac{1}{\sqrt{K_B(\qu, \qu)}}\sum_{n=0}^{\infty}\Phi_n\Phi_n(\oqu)=e^{-\frac{|\qu|^2}{2}}\sum_{n=0}^{\infty}\Phi_n\frac{\qu^n}{n!}
=e^{-\frac{|\qu|^2}{2}}\sum_{n=0}^{\infty}\frac{\qu^n}{n!}\cdot\Phi_n,
\end{equation}
where we have used the fact in Proposition \ref{lft_mul} (f), with a resolution of the identity
\begin{equation}\label{resolution}
\int_{\quat}|\eta_{\qu}\rangle\langle\eta_{\qu}|d\zeta(r, \theta, \phi, \psi)=I_{\HI^B_r}.
\end{equation}
Now take the corresponding annihilation and creation operators as
\begin{eqnarray*}
\as\Phi_0&=&0\\
\as\Phi_n&=&\sqrt{n}\Phi_{n-1}\\
\asd\Phi_n&=&\sqrt{n+1}\Phi_{n+1}.
\end{eqnarray*}
The operators can be taken as $\asd=\qu$ (multiplication by $\qu$) and $\as=\partial_s$ (left slice regular derivative), see \cite{Thi1, MuTh}. It is also not difficult to see that $(\asd)^{\dagger}=\as$, $[\as,\asd]=I_{\HI^B_r}$ and $\as\eta_{\qu}=\qu\cdot\eta_{\qu}$ (see also \cite{MTI}).\\

First of all, as in the complex quantum mechanics, all the operators considered here are unbounded operators. However, the operators act as $\HIB\ni|\phi\rangle\mapsto |\psi\rangle\in\HIB$, that is, the domain and the range of the operators are dense subsets of $\HI$. Furthermore, the Hilbert space, $\HIB$, can be taken as a space right-spanned by the regular functions $\{\frac{\qu^m}{m!}~~|~~m\in\mathbb{N}\}$ or anti-regular functions $\{\frac{\oqu^m}{m!}~~|~~m\in\mathbb{N}\}$ over $\quat$ (counterparts of holomorphic and anti-holomorphic functions). In this respect, the operators considered here do not have any domain problems as for the operators in the complex quantum mechanics. Therefore, we can use the operator tools of complex quantum mechanics, in particular, the Baker-Campbell-Hausdorff formula (for a complex argument along these lines see chapter 14 in \cite{Brian}).

The following Proposition demonstrate commutativity between quaternions and the right linear operators $\as$ and $\asd$. Further, it plays an important role.
\begin{proposition}\label{xAq}\cite{MTI}
	For each $\mathfrak{q}\in\quat$, we have $\mathfrak{q}\cdot {\mathsf{a}} ={\mathsf{a}} \cdot\mathfrak{q}$ and $\mathfrak{q}\cdot {{\mathsf{a}}^\dagger} ={{\mathsf{a}}^\dagger} \cdot\mathfrak{q}$.
\end{proposition}
%%%%%%%%%%%%%%%%%%%%%%%%%%%%%%%%%%%%%%%%%%%%%%%%%%%%%%%%%%%%%%%%%%%%%%%%%%%%%%%%%%%%%%%
\subsection{The right quaternionic displacement operator}
On a right quaternionic Hilbert space with a right multiplication we cannot have a displacement operator as a representation for the representation space $\HI^B_r$. This fact has been indicated twice in the literature, in \cite{Ad2} while studying quaternionic Perelomov type CS and in \cite{Thi2} when the authors studied the quaternionic canonical CS. However, in \cite{MTI}, we have shown that if we consider a right quaternionic Hilbert space with a left multiplication on it, see Eq. (\ref{lft_mul-op}), we can have a displacement operator as a representation for the representation space $\HI^B_r$ with all the desired properties. We shall extract some materials from \cite{MTI} as needed here.\\

\begin{proposition}\label{dis}\cite{MTI} The right quaternionic displacement operator $\D(\qu)=e^{\qu\cdot\asd-\oqu\cdot\as}$ is a unitary, square integrable and irreducible representation of the representation space $\HI^B_r$.
\end{proposition}

The following proposition discusses two versions for the displacement operator, the complex version is commonly used in complex quantum mechanics without hesitation, namely normal and anti-normal orderings.
\begin{proposition}\cite{MTI}
	The displacement operator $\D(\qu)$ satisfies
	\begin{itemize}
		\item [(i)] the normal ordering property:  $\D(\qu)=e^{-\frac{|\qu|^2}{2}}e^{\qu\cdot\mathsf{a}^\dagger}e^{-\oqu\cdot\mathsf{a}}$,
		\item [(ii)] the anti-normal ordering property:  $\D(\qu)=e^{\frac{|\qu|^2}{2}}e^{-\oqu\cdot\mathsf{a}}e^{\qu\cdot\mathsf{a}^\dagger}$.
	\end{itemize}
	Furthermore, the coherent state $\eta_{\qu}$ is generated from the ground state $\Phi_0$ by the displacement operator $\D(\qu)$,
	\begin{equation}\label{coh_dis}
	\eta_{\qu}=\D(\qu)\Phi_0.
	\end{equation}
\end{proposition}
\begin{proposition}\label{dis}\cite{MTI} The displacement operator $\D(\qu)$ satisfies the following properties
$$(i)~~\D(\qu)^\dagger\as\D(\qu)=\as+\qu\quad (ii)~~\D(\qu)^\dagger\asd\D(\qu)=\asd+\oqu.$$
\end{proposition}

%%%%%%%%%%%%%%%%%%%%%%%%%%%%%%%%%%%%%%%%%%%%%%%%%%%%%%%
\subsection{The right quaternionic squeezed operator}
Same reason as for the displacement operator, with a right multiplication on a right quaternionic Hilbert space the squeezed operator cannot be unitary. However, it becomes unitary with a left multiplication on a right quaternionic Hilbert space.
\begin{lemma}
The operator $A=\pu\cdot (\asd)^2-\overline{\pu}\cdot\as^2$ is anti-hermitian.
\end{lemma}
\begin{proof}
Consider
\begin{eqnarray*}
A^{\dagger}&=&(\pu\cdot (\asd)^2-\overline{\pu}\cdot\as^2)^{\dagger}\\
&=&((\asd)^2)^{\dagger}\cdot\overline{\pu}-(\as^2)^{\dagger}\cdot\pu\quad {\text{by}}~~\ref{sc_mul_aj-op}\\
&=&\as^2\cdot\overline{\pu}-(\asd)^2\cdot\pu\\
&=&\overline{\pu}\cdot \as^2-\pu\cdot (\asd)^2\quad \text{by Prop.} \ref{xAq}\\
&=&-A.
\end{eqnarray*}
\end{proof}
Let $A^\dagger=-A=B$, then $A$ and $B$ commute and both commute with the commutator $[A,B]$. Further $e^{-\frac{1}{2}[A, B]}=1$, therefore by the Baker-Campbell-Hausdorff formula,
$$e^{A}e^{B}e^{-\frac{1}{2}[A, B]}=e^{A+B}$$
we have, for the operator
$$S(\pu)=e^{\frac{1}{2}(\pu\cdot (\asd)^2-\overline{\pu}\cdot\as^2)},$$
$$S(\pu)S(\pu)^\dagger=e^{\frac{1}{2}A}e^{\frac{1}{2}A^\dagger}=e^{\frac{1}{2}(A-A)}=I_{\HI^B_r}.$$
That is the operator $S(\pu)$ is unitary and we call this operator the {\em quaternionic squeeze operator}. Further
$$S(\pu)^\dagger=e^{-\frac{1}{2}A}=S(-\pu).$$
If we take
$$K_+=\frac{1}{2}(\asd)^2,\quad K_{-}=\frac{1}{2}\as^2,\quad\text{and}\quad K_0=\frac{1}{2}(\asd\as+\frac{1}{2}I_{\HI^B_r}),$$
Then they satisfy the commutation rules
$$[K_0, K_+]=K_+,\quad [K_0, K_{-}]=-K_{-},\quad\text{and}\quad [K_+,K_{-}]=-2K_0.$$
That is, $K_+, K_{-}$ and $K_0$ are the generators of the $su(1,1)$ algebra and they satisfy the $su(1,1)$ commutation rules. In terms of these operators the squeeze operator $S(\pu)$ can be written as
\begin{equation}\label{sq1}
S(\pu)=e^{\pu\cdot K_+-\opu\cdot K_{-}}.
\end{equation}
%%%%%%%%%%%%%%%%%%%%%%%%%%%%%%%%%%%%%%%%%%%%%%%%%%%%%%%%%%%%%%%%%%%%%%%%%%%%%%%%%%%%%%%%%%%%%%%%
\subsection{Some Quaternionic Lie Algebras}
In the complex quantum mechanics the generators $\{N, \as^2, (\asd)^2\}$ spans the Lie algebra $su(1,1)$ and this algebra is involved in the construction of pure squeezed states.  The generators $\{\as, \asd, I, N, \as^2, (\asd)^2\}$ spans a six dimensional algebra $\mathfrak{h}_6$, which is involved in the construction of the generalized squeezed states \cite{Gaz}. In the following we generalize it to quaternions.\\

Let $\tau\in\{i,j,k\}$ and define
$$\mathfrak{h}_{6}^{(\tau)}=\mbox{linear span over }\mathbb{C}_\tau\,\{\mathbb{I}_{\HIB}, \mathsf{a}, \mathsf{a}^\dagger,N=\mathsf{a}^\dagger\mathsf{a}, \mathsf{a}^2,(\mathsf{a}^\dagger)^2\};$$
where $\mathbb{C}_\tau=\{x=x_1+\tau x_2~|~x_1, x_2\in\mathbb{R}\}$. Then the Proposition \ref{lft_mul} guarantees, together with the Remark 2.5, that $\mathfrak{h}_{6}^{(\tau)}$ is a vector space over $\mathbb{C}_\tau$ under the left multiplication \textquoteleft$\cdot$\textquoteright \,which is defined in (\ref{lft_mul-op}). Define \begin{equation}\label{lb1}
[\cdot,\cdot]_\tau:\mathfrak{h}_{6}^{(\tau)}\times\mathfrak{h}_{6}^{(\tau)}
\longrightarrow\mathfrak{h}_{6}^{(\tau)} \quad{\text{by}}\quad
[\mathcal{A}, \mathcal{B}]_\tau=\mathcal{A} \mathcal{B}-\mathcal{B} \mathcal{A},\mbox{~~for all~~}\mathcal{A}, \mathcal{B}\in\mathfrak{h}_{6}^{(\tau)}.
\end{equation}
One can easily see that the bracket $[\cdot, \cdot]$ satisfies the following axioms:
\begin{itemize}
	\item[(a)] \textit{Bilinearity}: for all $x,y\in\mathbb{C}_\tau$ and $\mathcal{A},\mathcal{B},\mathcal{C}\in\mathfrak{h}_{6}^{(\tau)}$, $$[x\mathcal{A}+y\mathcal{B},\mathcal{C}]_\tau=x[\mathcal{A},\mathcal{C}]_\tau+y[\mathcal{B},\mathcal{C}]_\tau
	\quad{\text{and}}\quad [\mathcal{A},x\mathcal{B}+y\mathcal{C}]_\tau=x[\mathcal{A},\mathcal{B}]_\tau+y[\mathcal{A},\mathcal{C}]_\tau.$$
	\item [(b)]\textit{Alternativity}: $[\mathcal{A},\mathcal{A}]_\tau=0$, for all  $\mathcal{A}\in\mathfrak{h}_{6}^{(\tau)}$.
	\item [(c)] \textit{The Jacobi identity}: for all  $\mathcal{A},\mathcal{B},\mathcal{C}\in\mathfrak{h}_{6}^{(\tau)}$. $$[\mathcal{A},[\mathcal{B},\mathcal{C}]_\tau]_\tau+[\mathcal{C},
	[\mathcal{A},\mathcal{B}]_\tau]_\tau+[\mathcal{B},[\mathcal{C},\mathcal{A}]_\tau]_\tau=0.$$
	\item [(d)] \textit{Anti-commutativity}: $[\mathcal{A},\mathcal{B}]_\tau=-[\mathcal{B},\mathcal{A}]_\tau$, for all  $\mathcal{A},\mathcal{B}\in\mathfrak{h}_{6}^{(\tau)}$.
\end{itemize}
Let $\mathcal{A}, \mathcal{B}\in\mathfrak{h}_{6}^{(\tau)}$, then there exists $a,b,c,d,e,f,u,v,w,x,y,z\in\mathbb{C}_\tau$ such that
$$\mathcal{A}=a\cdot\mathbb{I}_{\HIB}+ b\cdot\mathsf{a}+  c\cdot \mathsf{a}^\dagger+ d\cdot N+ e\cdot\mathsf{a}^2+f\cdot(\mathsf{a}^\dagger)^2$$
and $$\mathcal{B}=u\cdot\mathbb{I}_{\HIB}+ v\cdot\mathsf{a}+  w\cdot \mathsf{a}^\dagger+ x\cdot N+ y\cdot\mathsf{a}^2+z\cdot(\mathsf{a}^\dagger)^2.$$
Using the facts that
\begin{center}
	\begin{tabular}{c c c c}
		$[\mathsf{a},\mathsf{a}^\dagger]_\tau=\mathbb{I}_{\HIB}$, &	$[\mathsf{a},N]_\tau=\mathsf{a}$, & $[\mathsf{a}^\dagger,N]_\tau=-\mathsf{a}^\dagger$, & $[\mathsf{a}^2,(\mathsf{a}^\dagger)^2]_\tau=-2(2N+\mathbb{I}_{\HIB}),$\\\\
		$[\mathsf{a}^2,\mathsf{a}^\dagger]_\tau=2\mathsf{a}$, & $[(\mathsf{a}^\dagger)^2,\mathsf{a}]_\tau=-2\mathsf{a}^\dagger$, & $[\mathsf{a}^2,N]_\tau=2\mathsf{a}^2$,& $[(\mathsf{a}^\dagger)^2,N]_\tau=-2(\mathsf{a}^\dagger)^2$.
	\end{tabular}
\end{center}
with the aid of Proposition \ref{xAq}, we can obtain that
$[\mathcal{A},\mathcal{B}]_\tau\in\mathfrak{h}_{6}^{(\tau)}$. Hence $\mathfrak{h}_{6}^{(\tau)}$ is a Lie algebra with the Lie bracket $[\cdot,\cdot]_\tau$. It is a sub case of the following Lie algebra, since $\mathfrak{h}_{6}^{(\tau)}$ involves a single $\tau\in\{i,j,k\}$ at a time.

One can easily check that the subset (but it is a linear space itself over $\mathbb{R}$) of $\mathfrak{h}_{6}^{(\tau)}$, $$\mbox{linear span over }\mathbb{R}\,\{\mathbb{I}_{\HIB}, \mathsf{a}, \mathsf{a}^\dagger,N=\mathsf{a}^\dagger\mathsf{a}, \mathsf{a}^2,(\mathsf{a}^\dagger)^2\}$$ forms a Lie algebra with the Lie bracket $[\mathcal{A},\mathcal{B}]_1=\mathcal{A} \mathcal{B}-\mathcal{B} \mathcal{A}$, for all elements $\mathcal{A},\mathcal{B}$ in this linear space (it is a restriction map of $[\cdot,\cdot]_\tau$). Furthermore, we have another subset  $$\mathfrak{h}_{12}^{(\tau)}=\mbox{linear span over }\mathbb{R}\,\{\mathbb{I}_{\HIB},  N,\mathsf{a}, \mathsf{a}^\dagger, \mathsf{a}^2,(\mathsf{a}^\dagger)^2,\tau\cdot\mathbb{I}_{\HIB},  \tau\cdot N, \tau\cdot\mathsf{a},\tau\cdot\mathsf{a}^\dagger,\tau\cdot\mathsf{a}^2,\tau\cdot(\mathsf{a}^\dagger)^2 \}$$ is a linear space over $\mathbb{R}$, and forms Lie algebra with the Lie bracket $[\cdot,\cdot]_\tau$.
Moreover, An arbitrary element $\mathcal{A}\in\mathfrak{h}_{12}^{(\tau)}$ takes the form of
\begin{equation}\label{eqA}
\mathcal{A}=a_1\cdot\mathbb{I}_{\HIB}+  a_2\cdot N+a_3\cdot\mathsf{a}+a_4\cdot\mathsf{a}^\dagger+a_5\cdot\mathsf{a}^2+a_6\cdot(\mathsf{a}^\dagger)^2+a_\tau^{(1)} \tau\cdot\mathsf{a}+a_\tau^{(2)}\tau\cdot\mathsf{a}^\dagger +a_\tau^{(3)}\tau\cdot\mathsf{a}^2+a_\tau^{(4)}\tau\cdot(\mathsf{a}^\dagger)^2;
\end{equation}
where $a_l,a_\tau^{(m)}\in\mathbb{R},$ for all $l=1,2,\cdots,6$, $m=1,2,\cdots,4$.
It can be simply expressed as
\begin{equation}\label{Expr}
\mathcal{A}=\mathcal{A}_1+\mathcal{A}_\tau;
\end{equation}
where
\begin{equation}\label{eqA1}
\mathcal{A}_1=a_1\cdot\mathbb{I}_{\HIB}+  a_2\cdot N+a_3\cdot\mathsf{a}+a_4\cdot\mathsf{a}^\dagger+a_5\cdot\mathsf{a}^2+a_6\cdot(\mathsf{a}^\dagger)^2,
\end{equation}
\begin{equation}\label{eqAtau}
\mathcal{A}_\tau=a_\tau^{(1)} \tau\cdot\mathbb{I}_{\HIB}+  a_\tau^{(2)} \tau\cdot N+a_\tau^{(3)} \tau\cdot\mathsf{a}+a_\tau^{(4)}\tau\cdot\mathsf{a}^\dagger +a_\tau^{(5)}\tau\cdot\mathsf{a}^2+a_\tau^{(6)}\tau\cdot(\mathsf{a}^\dagger)^2.
\end{equation}

Only for notational convenience, in order to write a quaternion as $\qu=q_0+\sum_{\tau=i,j,k}q_{\tau}\tau$, we shall write $\qu=q_0+q_ii+q_jj+q_kk$ with $q_0,q_i,q_j,q_k\in\mathbb{R}$. Let $$\mathfrak{h}_{24}=\mbox{linear span over }\,\mathbb{R}\,\{\tau\cdot\mathbb{I}_{\HIB},\tau\cdot N, \tau\cdot\mathsf{a},\tau\cdot\mathsf{a}^\dagger,\tau\cdot\mathsf{a}^2,\tau\cdot(\mathsf{a}^\dagger)^2~|~\tau=1,i,j,k\}.$$
Then $\mathfrak{h}_{24}$ is a vector space over $\mathbb{R}$, and it contains  $\mathfrak{h}_{12}^{(\tau)}$. Define the map $[\cdot,\cdot]:\mathfrak{h}_{24}\times\mathfrak{h}_{24}\longrightarrow\mathfrak{h}_{24},$ by the setting (using the expression (\ref{Expr}))
\begin{equation}\label{liebraH}
[\mathcal{A},\mathcal{B}]:=[\mathcal{A}_1,\mathcal{B}_1]_1+\sum_{\tau=i,j,k}[\mathcal{A}_1,\mathcal{B}_\tau]_\tau+\sum_{\tau=i,j,k}[\mathcal{A}_\tau,\mathcal{B}_1+\mathcal{B}_\tau]_\tau,\mbox{~~for all ~~}\mathcal{A},\mathcal{B}\in\mathfrak{h}_{24}.
\end{equation}
Alternatively it can be written as
\begin{equation}\label{liebraHAl}
[\mathcal{A},\mathcal{B}]=\sum_{\tau=i,j,k}\left(\dfrac{1}{3} [\mathcal{A}_1,\mathcal{B}_1]_\tau+[\mathcal{A}_1,\mathcal{B}_\tau]_\tau+[\mathcal{A}_\tau,\mathcal{B}_1+\mathcal{B}_\tau]_\tau\right),\mbox{~~for all ~~}\mathcal{A},\mathcal{B}\in\mathfrak{h}_{24}.
\end{equation}
Since $[\cdot,\cdot]_\tau$ is a Lie bracket, it follows that $\mathfrak{h}_{24}$ is a Lie algebra with the Lie bracket $[\cdot,\cdot]$. The proof of following Proposition follows using this Lie bracket $[\cdot,\cdot]$.
\begin{proposition}\label{Psqop}
Let $\displaystyle\pu=|\pu|e^{i\theta\sigma(\hat{n})}$ and $N=\asd\as$, the number operator, then the squeeze operator $S(\pu)$ satisfies the following relations
\begin{eqnarray*}
(i)~S(\pu)^{\dagger}\as S(\pu)&=&(\cosh{|\pu|})\as+\left(e^{i\theta\sigma(\hat{n})}\sinh{|\pu|}\right)\cdot\asd.\\
(ii)~S(\pu)^{\dagger}\asd S(\pu)&=&(\cosh{|\pu|})\asd+\left(e^{-i\theta\sigma(\hat{n})}\sinh{|\pu|}\right)\cdot\as.\\
(iii)~S(\pu)^{\dagger}N S(\pu)&=&(\cosh^2{|\pu|})\asd\as+\left(e^{-i\theta\sigma(\hat{n})}\sinh{|\pu|}\cosh{|\pu|}\right)\cdot\as^2\\
&+&\left(e^{i\theta\sigma(\hat{n})}\sinh{|\pu|}\cosh{|\pu|}\right)\cdot(\asd)^2+(\sinh^2{|\pu|})\as\asd.
\end{eqnarray*}
\end{proposition}
\begin{proof}
With $A=\frac{1}{2}(\pu\cdot(\asd)^2-\opu\cdot\as^2)$ and the commutation rule $[\as,\asd]=I_{\HIB}$ we can calculate
\begin{eqnarray*}
& &[-A, \as]=\pu\cdot\asd\\
& &[-A,[-A, \as]]=|\pu|^2\as\\
& & [-A, [-A,[-A, \as]]]=|\pu|^2\pu\cdot\asd\\
& &[-A, [-A, [-A,[-A, \as]]]]=|\pu|^4\as\\
& &[-A, [-A, [-A, [-A,[-A, \as]]]]]=|\pu|^4\pu\cdot\asd\\
& &\cdots\cdots\cdots
\end{eqnarray*}
Therefore, by using the identity $e^CBe^{-C}=B+[C,B]+\frac{1}{2!}[C,[C,B]]+\cdots$ we have
\begin{eqnarray*}
S(\pu)^{\dagger}\as S(\pu)&=&\as+\pu\cdot\asd+\frac{1}{2!}|\pu|^2\as+\frac{1}{3!}|\pu|^2\pu\cdot\asd+\frac{1}{4!}|\pu|^4\as
+\frac{1}{5!}|\pu|^4\pu\cdot\asd +\cdots\\
&=&(\as+\frac{1}{2!}|\pu|^2\as+\frac{1}{4!}|\pu|^4\as+\cdots)+e^{i\theta\sigma(\hat{n})}\cdot
(|\pu|\asd+\frac{1}{3!}|\pu|^3\asd+\frac{1}{5!}|\pu|^5\asd+\cdots)\\
&=&\left(\sum_{n=0}^\infty\frac{|\pu|^{2n}}{(2n)!}\right)\as+e^{i\theta\sigma(\hat{n})}
\left(\sum_{n=0}^\infty\frac{|\pu|^{2n+1}}{(2n+1)!}\right)\cdot\asd\\
&=&(\cosh|\pu|)\as+\left(e^{i\theta\sigma(\hat{n})}\sinh|\pu|\right)\cdot\asd.
\end{eqnarray*}
The second relation is the hermitian conjugate of the first one. The third relation can be obtained by writing
$$S(\pu)^{\dagger}N S(\pu)=S(\pu)\asd\as S(\pu)=S(\pu)\as S(\pu)S(\pu)^{\dagger}\as S(\pu)$$
and then multiplying the first and the second relations.
\end{proof}
%%%%%%%%%%%%%%%%%%%%%%%%%%%%%%%%%%%%%%%%%%%%%%%%%%
\subsection{Right quaternionic quadrature operators}
We introduce the quadrature operators analogous to the complex quadrature operators with a left multiplication on a right quaternionic Hilbert space.
\begin{equation}\label{Q-op}
X=\frac{1}{2}(\as+\asd)\quad\text{and}\quad Y=-\frac{i}{2}\cdot (\as-\asd),
\end{equation}
where the quaternion unit $i$ in $Y$ can be replaced by $j, k$ or any $I\in\mathbb{S}$ (see \cite{MTI}).
\begin{proposition}
The operators $X$ and $Y$ are self-adjoint and $[X,Y]=\frac{i}{2}\cdot I_{\HI^B_r}.$
\end{proposition}
\begin{proof}
The self-adjointness of the operators  follows prom the Prop. \ref{xAq} (see \cite{MTI}). Now it is straight forward to compute, with the aid of Prop. \ref{xAq},
$$XY=-\frac{1}{4}(\as+\asd)(i\cdot (\as-\asd)=-\frac{1}{4}i\cdot (\as^2-\as\asd+\asd\as-(\asd)^2)$$
and
$$YX=-\frac{1}{4}i\cdot(\as-\asd)(\as+\asd)=-\frac{1}{4}i\cdot(\as^2+\as\asd-\asd\as-(\asd)^2).$$
Thus $$[X,Y]=XY-YX=-\frac{1}{2}i\cdot(-\as\asd+\asd\as)=\frac{1}{2}i\cdot[\as,\asd]=\frac{1}{2}i\cdot I_{\HI^B_r}.$$
\end{proof}
%%%%%%%%%%%%%%%%%%%%%%%%%%%%%%%%%%%%%%%%%%%%%%%%%%%%%%
\begin{definition}\cite{Gaz} Let $A$ and $B$ be quantum observables with commutator $[A, B]=i\cdot C$.Then from Cauchy-Schwarz inequality $(\Delta A)(\Delta B)\geq\frac{1}{2}|\langle C\rangle|$. A state will be called squeezed with respect to the pair $(A, B)$ if $(\Delta A)^2$ (or $(\Delta B)^2)<\frac{1}{2}|\langle C\rangle|$.
A state is called ideally squeezed if the equality $(\Delta A)(\Delta B)=\frac{1}{2}|\langle C\rangle|$ is reached together with $(\Delta A)^2$ (or $(\Delta B)^2)<\frac{1}{2}|\langle C\rangle|$.
\end{definition}
We adapt the same definition for quaternionic squeezed states.
\subsection{Right quaternionic pure squeezed states}
A pure squeezed state is produced by the sole action of the unitary operator $S(\pu)$ on the vacuum state.
Let $C=\frac{1}{2}\pu\cdot(\asd)^2$ and $D=\frac{1}{2}\overline{\pu}\cdot\as^2$. Then it can be computed that
\begin{eqnarray*}
[C,D]\Phi_n&=&-\frac{1}{4}\left(\pu\cdot(\asd)^2~\overline{\pu}\cdot\as^2
-\overline{\pu}\cdot\as^2~\pu\cdot(\asd)^2\right)\Phi_n\\
&=&\frac{1}{4}|\pu|^2\left((\asd)^2\as^2-\as^2(\asd)^2\right)\Phi_n\\
&=&-\frac{1}{2}|\pu|^2(n+1)\Phi_n.
\end{eqnarray*}
That is $[C,D]=-\frac{1}{2}|\pu|^2(n+1)I_{\HI^B_r}$. Further, similarly, we can obtain
$$[C,[C,D]]\Phi_n=0\quad\text{and}\quad [D,[C,D]]\Phi_n=0.$$
That is $[C,[C,D]]=0$ and $[D,[C,D]]=0$. Therefore from the BCH formula we have
$$S(\pu)=e^{C-D}=e^{-\frac{1}{2}[C,D]}e^Ce^{-D}.$$
Now
\begin{eqnarray*}
S(\pu)\Phi_0&=&e^{-\frac{1}{2}[C,D]}e^Ce^{-D}\Phi_0\\
&=&e^{-\frac{1}{2}[C,D]}e^C\Phi_0\\
&=&e^{-\frac{1}{2}[C,D]}\sum_{n=0}^{\infty}\frac{(\pu\cdot(\asd)^2)^n}{2^n n!}\Phi_0\\
&=&e^{-\frac{1}{2}[C,D]}\sum_{n=0}^{\infty}\frac{\pu^n\sqrt{(2n)!}}{2^n n!}\cdot\Phi_{2n}.
\end{eqnarray*}
Further,
\begin{eqnarray*}
[C,D]\Phi_{2n}&=&(\frac{1}{2}\pu\cdot(\asd)^2 \frac{1}{2}\opu\cdot \as^2-\frac{1}{2}\opu\cdot\as^2 \frac{1}{2}\pu(\asd)^2)\Phi_{2n}\\
&=&\frac{1}{4}|\pu|^2((\asd)^2\as^2-\as^2(\asd)^2)\Phi_{2n}\\
&=&\frac{1}{4}|\pu|^2(-8n-2)\Phi_{2n}
\end{eqnarray*}
Therefore
$$e^{-\frac{1}{2}[C,D]}\Phi_{2n}=e^{\frac{1}{8}|\pu|^2(8n+2)}\Phi_{2n}=e^{n|\pu|^2}e^{\frac{1}{4}|\pu|^2}\Phi_{2n}.$$
Thus
\begin{equation}\label{SCS}
S(\pu)\Phi_0=\eta_\pu=e^{\frac{1}{4}|\pu|^2}\sum_{n=0}^{\infty}e^{n|\pu|^2}\frac{\pu^n\sqrt{(2n)!}}{2^n n!}\cdot\Phi_{2n}.
\end{equation}
Since $S(\pu)$ is a unitary operator, by construction we have
$$\langle\eta_\pu|\eta_\pu\rangle=\langle S(\pu)\Phi_0| S(\pu)\Phi_0\rangle=\langle\Phi_0|\Phi_0\rangle=1.$$
The states $\eta_\pu$ are normalized. Since the pure squeezed state $\eta_\pu$ only possess the even numbered basis vector, $\{\Phi_{2n}~~|~~n=0,1,2,\cdots\}$ a resolution of the identity cannot hold on $\HIB$. However, if we form a space right spanned by $\{\Phi_{2n}~~|~~n=0,1,2,\cdots\}$ over the quaternions it may be possible to find a resolution of the identity for that space. However, such an attempt is not necessary and even in the complex case, as far as we know, it does not exist in the literature.
%%%%%%%%%%%%%%%%%%%%%%
\subsubsection{Expectation values and the variances}
For a normalized state $\eta$ the expectation value of an operator $F$ is $\langle F\rangle=\langle\eta|F|\eta\rangle$. First let us see the expectation values of $\as$ and $\asd$ using Proposition \ref{Psqop}.
\begin{eqnarray*}
\langle\as\rangle&=&\langle\eta_\pu|\as|\eta_\pu\rangle=\langle S(\pu)\Phi_0|\as|S(\pu)\Phi_0\rangle\\
&=&\langle\Phi_0|S(\pu)^\dagger\as S(\pu)\Phi_0\rangle\\
&=&\langle\Phi_0|(\cosh{|\pu|})\as+\left(e^{i\theta\sigma(\hat{n})}\sinh{|\pu|}\right)\cdot\asd\Phi_0\rangle\\
&=&(\cosh{|\pu|})\langle\Phi_0|\as\Phi_0\rangle+\sinh{|\pu|}\langle\Phi_0|
\left(e^{i\theta\sigma(\hat{n})}\right)\cdot\asd\Phi_0\rangle\\
&=&0+ \sinh{|\pu|}\langle\Phi_0|
\left(e^{i\theta\sigma(\hat{n})}\right)\cdot\Phi_1\rangle\\
&=&\sinh{|\pu|}\langle\Phi_0|
\Phi_1e^{i\theta\sigma(\hat{n})}\rangle\quad\text{as}~~\Phi_1~~\text{is a basis vector, see Prop.\ref{lft_mul} (f)}\\
&=&\sinh{|\pu|}\langle\Phi_0|
\Phi_1\rangle e^{i\theta\sigma(\hat{n})}=0\\
\end{eqnarray*}
Similarly we get
$$\langle\asd\rangle=\langle\eta_\pu|\asd|\eta_\pu\rangle=0.$$
Hence we get
$$\langle X\rangle=\langle\eta_\pu|X|\eta_\pu\rangle=0\quad\text{and}\quad \langle Y\rangle=\langle\eta_\pu|Y|\eta_\pu\rangle=0.$$
Since
\begin{eqnarray*}
& &S(\pu)^\dagger\as S(\pu)S(\pu)^{\dagger}\asd S(\pu)=S(\pu)^\dagger\as\asd S(\pu)\\
&=&\cosh^2{|\pu|}\as\asd+e^{-i\theta\sigma(\hat{n})}\cosh{|\pu|}\sinh{|\pu|}\cdot\as^2
+e^{i\theta\sigma(\hat{n})}\cosh{|\pu|}\sinh{|\pu|}\cdot(\asd)^2+\sinh^2{|\pu|}\asd\as
\end{eqnarray*}
and similarly
\begin{eqnarray*}
& &S(\pu)^\dagger\asd\as S(\pu)\\
&=&\cosh^2{|\pu|}\asd\as+e^{-i\theta\sigma(\hat{n})}\cosh{|\pu|}\sinh{|\pu|}\cdot\as^2
+e^{i\theta\sigma(\hat{n})}\cosh{|\pu|}\sinh{|\pu|}\cdot(\asd)^2+\sinh^2{|\pu|}\as\asd
\end{eqnarray*}
\begin{eqnarray*}
S(\pu)^\dagger\as^2 S(\pu)
&=&\cosh^2{|\pu|}\as^2+e^{i\theta\sigma(\hat{n})}\cosh{|\pu|}\sinh{|\pu|}\cdot(\as\asd+\asd\as)
+e^{2i\theta\sigma(\hat{n})}\sinh^2{|\pu|}\cdot(\asd)^2\\
S(\pu)^\dagger\as^2 S(\pu)
&=&\cosh^2{|\pu|}(\asd)^2+e^{-i\theta\sigma(\hat{n})}\cosh{|\pu|}\sinh{|\pu|}\cdot(\as\asd+\asd\as)
+e^{-2i\theta\sigma(\hat{n})}\sinh^2{|\pu|}\cdot\as^2.
\end{eqnarray*}
Using the above relations we readily obtain
\begin{eqnarray*}
\langle\as\asd\rangle&=&\langle\eta_\pu|\as\asd|\eta_\pu\rangle=\cosh^2{|\pu|}\\
\langle\asd\as\rangle&=&\langle\eta_\pu|\asd\as|\eta_\pu\rangle=\sinh^2{|\pu|}\\
\langle\as^2\rangle&=&\langle\eta_\pu|\as^2|\eta_\pu\rangle=\cosh{|\pu|}\sinh{|\pu|}e^{i\theta\sigma(\hat{n})}\\
\langle(\asd)^2\rangle&=&\langle\eta_\pu|(\asd)^2|\eta_\pu\rangle=\cosh{|\pu|}\sinh{|\pu|}e^{-i\theta\sigma(\hat{n})}.
\end{eqnarray*}
Since $X^2=\frac{1}{4}((\asd)^2+\as\asd+\asd\as+\as^2)$ and $Y^2=\frac{1}{4}(\as\asd+\asd\as-\as^2-(\asd)^2)$ we have
\begin{eqnarray*}
\langle\eta_\pu|X^2|\eta_\pu\rangle&=&\frac{1}{4}\left\{(\cosh^2{|\pu|}\sinh^2{|\pu|})\mathbb{I}_2
+\cosh{|\pu|}\sinh{|\pu|}\left(e^{i\theta\sigma(\hat{n})}+e^{-i\theta\sigma(\hat{n})}\right)\right\}\\
&=&\frac{1}{4}\left\{\cosh(2|\pu|)\mathbb{I}_2+\sinh(2|\pu|)\cos{(\theta\sigma(\hat{n}))}\right\}\\
\langle\eta_\pu|Y^2|\eta_\pu\rangle
&=&\frac{1}{4}\left\{\cosh(2|\pu|)\mathbb{I}_2-\sinh(2|\pu|)\cos{(\theta\sigma(\hat{n}))}\right\}
\end{eqnarray*}
Since $\langle\Delta X\rangle^2=\langle\eta_\pu|X^2|\eta_\pu\rangle-\langle\eta_\pu|X|\eta_\pu\rangle^2$ and $\langle\eta_\pu|X|\eta_\pu\rangle=0$ we have
\begin{eqnarray*}
\langle\Delta X\rangle^2&=&\frac{1}{4}\left\{\cosh(2|\pu|)\mathbb{I}_2+\sinh(2|\pu|)\cos{(\theta\sigma(\hat{n}))}\right\}\\
\langle\Delta Y\rangle^2&=&\frac{1}{4}\left\{\cosh(2|\pu|)\mathbb{I}_2-\sinh(2|\pu|)\cos{(\theta\sigma(\hat{n}))}\right\}
\end{eqnarray*}
Hence
\begin{eqnarray*}
\langle\Delta X\rangle^2\langle\Delta Y\rangle^2&=&
\frac{1}{16}\left\{\cosh^2(2|\pu|)\mathbb{I}_2-\sinh^2(2|\pu|)\cos^2{(\theta\sigma(\hat{n}))}\right\}\\
&=&\frac{1}{16}\left\{\cosh^2(2|\pu|)\mathbb{I}_2-\sinh^2(2|\pu|)(1-\sin^2{(\theta\sigma(\hat{n}))})\right\}\\
&=&\frac{1}{16}\left\{\mathbb{I}_2+\sinh^2(2|\pu|)\sin^2{(\theta\sigma(\hat{n}))}\right\}
\end{eqnarray*}
An exact analogue of the complex case. Since we are in the quaternions, it appears as a $2\times 2$ matrix. Further in the complex case, the product of the variances depends on $r$ and $\theta$ (when $z=re^{i\theta}$). In the quaternion case it depends on all four parameters $r, \theta, \phi$ and $\psi$. Let us write
$$\displaystyle U+iV=e^{-\frac{i}{2}\theta\sigma(\hat{n})}\cdot(X+iY)=e^{-\frac{i}{2}\theta\sigma(\hat{n})}\cdot\as.$$
Then using Proposition \ref{xAq} we can write
\begin{eqnarray*}
S(\pu)^{\dagger}(U+iV)S(\pu)&=&e^{-\frac{i}{2}\theta\sigma(\hat{n})}\cdot S(\pu)^\dagger\as S(\pu)\\
&=&e^{-\frac{i}{2}\theta\sigma(\hat{n})}\cdot(\cosh{|\pu|}\as+e^{i\theta\sigma(\hat{n})}\sinh{|\pu|}\asd)\\
&=&\frac{e^{|\pu|}+e^{-|\pu|}}{2}e^{-\frac{i}{2}\theta\sigma(\hat{n})}\cdot\as
+\frac{e^{|\pu|}-e^{-|\pu|}}{2}e^{\frac{i}{2}\theta\sigma(\hat{n})}\cdot\asd\\
&=&\frac{1}{2}(e^{-\frac{i}{2}\theta\sigma(\hat{n})}\cdot\as+e^{\frac{i}{2}\theta\sigma(\hat{n})}\cdot\asd)e^{|\pu|}
+\frac{1}{2}(e^{-\frac{i}{2}\theta\sigma(\hat{n})}\cdot\as-e^{\frac{i}{2}\theta\sigma(\hat{n})}\cdot\asd)e^{-|\pu|}\\
&=&Ue^{|\pu|}+i\cdot Ve^{-|\pu|},
\end{eqnarray*}
with
\begin{eqnarray*}
U&=&\frac{1}{2}(e^{-\frac{i}{2}\theta\sigma(\hat{n})}\cdot\as+e^{\frac{i}{2}\theta\sigma(\hat{n})}\cdot\asd)e^{|\pu|}
\quad\text{and}\\
V&=&\frac{-i}{2}(e^{-\frac{i}{2}\theta\sigma(\hat{n})}\cdot\as-e^{\frac{i}{2}\theta\sigma(\hat{n})}\cdot\asd)e^{|\pu|}.
\end{eqnarray*}
Since
\begin{eqnarray*}
U^2&=&\frac{1}{4}(e^{-i\theta\sigma(\hat{n})}\cdot\as^2+\as\asd+\asd\as+e^{i\theta\sigma(\hat{n})}\cdot(\asd)^2)
\quad\text{and}\\
V^2&=&-\frac{1}{4}(e^{-i\theta\sigma(\hat{n})}\cdot\as^2-\as\asd-\asd\as+e^{i\theta\sigma(\hat{n})}\cdot(\asd)^2),
\end{eqnarray*}
it is straight forward that
$\langle\eta_\pu|U|\eta_\pu\rangle=0$, $\langle\eta_\pu|V|\eta_\pu\rangle=0$,
\begin{eqnarray*}
\langle\eta_\pu|U^2|\eta_\pu\rangle
&=&\frac{1}{4}(\cosh{|\pu|}\sinh{|\pu|}+\cosh^2{|\pu|}+\sinh^2{|\pu|}+\cosh{|\pu|}\sinh{|\pu|})\mathbb{I}_2\\
&=&\frac{1}{4}(\cosh{|\pu|}+\sinh{|\pu|})^2\mathbb{I}_2\quad \text{and}\\
\langle\eta_\pu|V^2|\eta_\pu\rangle
&=&-\frac{1}{4}(\cosh{|\pu|}\sinh{|\pu|}-\cosh^2{|\pu|}-\sinh^2{|\pu|}+\cosh{|\pu|}\sinh{|\pu|})\mathbb{I}_2\\
&=&\frac{1}{4}(\cosh{|\pu|}-\sinh{|\pu|})^2\mathbb{I}_2.
\end{eqnarray*}
Hence
$$\langle\Delta U\rangle^2 \langle\Delta V\rangle^2=\frac{1}{16}(\cosh^2{|\pu|}-\sinh^2{|\pu|})^2\mathbb{I}_2=\frac{1}{16}\mathbb{I}_2$$
and therefore
\begin{equation}
\langle\Delta U\rangle \langle\Delta V\rangle=\frac{1}{4}\mathbb{I}_2,
\end{equation}
while $\langle\Delta U\rangle \not=\langle\Delta V\rangle$, an exact analogue of the complex case \cite{Gaz}. Hence, the class of ideally squeezed states contains the set of quaternionic pure squeezed states.\\
Using the relation (iii) in Proposition \ref{Psqop} we obtain the mean photon number
$$\langle N\rangle=\langle\eta_\pu| N|\eta_\pu\rangle=\langle\Phi_0|S(\pu)^\dagger NS(\pu)\Phi_0\rangle=\sinh^2{|\pu|}\I_2.$$
Also using
$$\langle\eta_{\pu}| N^2|\eta_{\pu}\rangle=\langle\Phi_0|S(\pu)^\dagger N S(\pu)S(\pu)^\dagger N S(\pu)\Phi_0\rangle$$
we get
\begin{eqnarray*}
\langle N^2\rangle&=&\langle\Phi_0|S(\pu)^\dagger N S(\pu)S(\pu)^\dagger N S(\pu)\Phi_0\rangle\\
&=&(\sinh^4{|\pu|}+2\sinh^2{|\pu|}\cosh^2{|\pu|})\I_2\\
&=&3\sinh^4{|\pu|}+2\sinh^2{|\pu|}\I_2.
\end{eqnarray*}
Hence the variance is
$$\langle \Delta N\rangle^2=\langle N^2\rangle-\langle N\rangle^2=2\sinh^2{|\pu|}(1+\sinh^2{|\pu|})\I_2.$$
The photon number variance is also described by Mandel's Q-parameter. The Mandel parameter is \cite{Gaz, Lo}
\begin{eqnarray*}
Q_M&=&\frac{\langle \Delta N\rangle^2}{\langle N\rangle}-1\\
&=&\frac{2\sinh^2{|\pu|}(1+\sinh^2{|\pu|})}{\sinh^2{|\pu|}}\I_2-\I_2\\
&=&(1+2\sinh^2{|\pu|})\I_2=2\langle N\rangle+\I_2.
\end{eqnarray*}
Since $Q_M>0$ (as a positive definite matrix) the photon number probability distribution is super-Poissonian.
($Q_M=0$ Poissonian and $Q_M<0$ sub-Poissonian).
%%%%%%%%%%%%%%%%%%%%%%%%%%%%%%%%%%%%%%%%%%%%%%%%%%%%%%%%%%%%%%
\subsubsection{The pure squeezed states with anti-normal ordering of $S(\pu)$}
Even though in order to compute the expectation values and variances the relations in Proposition \ref{Psqop} enough, let us give an expression for the pure squeezed states with the anti-normal ordering of the operator $S(\pu)$.
\begin{eqnarray*}
\eta_{\pu}^a&=&S(\pu)\Phi_0\\
&=&e^{-\frac{1}{2}[C, D]}e^{-D}e^C\Phi_0\\
&=&e^{-\frac{1}{2}[C, D]}e^{-D}\sum_{m=0}^{\infty}\frac{\pu\cdot(\asd)^{2m}}{2^m}\Phi_0\\
&=&e^{-\frac{1}{2}[C, D]}e^{-D}\sum_{m=0}^{\infty}\frac{\pu^m\sqrt{(2m)!}}{2^m m!}\cdot\Phi_{2m}\\
&=&e^{-\frac{1}{2}[C, D]}\sum_{m=0}^{\infty}\sum_{n=0}^{\infty}\frac{\pu^m\sqrt{(2m)!}}{2^m m!}\frac{\opu^n\as^{2n}}{2^nn!}\cdot\Phi_{2m}\\
&=&e^{-\frac{1}{2}[C, D]}\sum_{m=0}^{\infty}\sum_{n=0}^{\infty}\frac{\pu^m\sqrt{(2m)!}}{2^m m!}\sqrt{\frac{(2m)!}{(2m-2n)!}}\frac{\opu^n}{2^nn!}\cdot\Phi_{2m}\\
&=&e^{-\frac{1}{2}[C, D]}\sum_{n=0}^{\infty}\sum_{m=n}^{\infty}\frac{\pu^m\opu^n(2m)!}{4^n m!n!\sqrt{(2m-2n)!}}\cdot\Phi_{2m-2n}\quad\text{assuming}~~m>n\\
&=&e^{-\frac{1}{2}[C, D]}\sum_{n=0}^{\infty}\sum_{s=0}^{\infty}\frac{\pu^{n+s}\opu^n(2n+2s)!}{4^n(n+s)!n!\sqrt{(2s)!} }\cdot\Phi_{2s}\quad\text{taking}~~m-n=s\\
&=&e^{\frac{1}{4}|\pu|^2}\sum_{n=0}^{\infty}\sum_{s=0}^{\infty}\frac{\pu^{n+s}\opu^n(2n+2s)!}{4^n(n+s)!n!\sqrt{(2s)!} }e^{s|\pu|^2}\cdot\Phi_{2s}.
\end{eqnarray*}

%%%%%%%%%%%%%%%%%%%%%%%%%%%%%%%%%%%%%
\subsection{Right quaternionic squeezed states}
In view of Prop. \ref{lft_mul}(f), for a basis vector $\qu\cdot\Phi_n=\Phi_n\qu$, therefore we write the canonical CS as
$$\eta_\qu=\D(\qu)\Phi_0=e^{-|\qu|^2/2}\sum_{n=0}^{\infty}\Phi_n\frac{\qu^n}{\sqrt{n!}}.$$
Let $\displaystyle S(\pu)\Phi_n=\Phi_n^{\pu}$, where the set $\{\Phi_n~~|~~n=0,1,2,\cdots\}$ is the basis of the Fock space of regular Bargmann space $\HI^{B}_r$. Since $S(\pu)$ is a unitary operator, the set $\{\Phi_n^{\pu}~~|~~n=0,1,2,\cdots\}$ is also form an orthonormal basis for $\HI^{B}_r$. That is
\begin{equation}\label{sqbasis}
\langle\Phi_m^{\pu}|\Phi_n^{\pu}\rangle=\delta_{mn}.
\end{equation}
Now the squeezed states are
\begin{equation}\label{sqstates}
\eta_{\qu}^{\pu}=S(\pu)\D(\qu)\Phi_0=S(\pu)\eta_{\qu}
=e^{-|\qu|^2/2}\sum_{n=0}^{\infty}\Phi_n^{\pu}\frac{\qu^n}{\sqrt{n!}}.
\end{equation}
Since the canonical CS are normalized, that is $\langle\eta_\qu|\eta_\qu\rangle=1$, and the squeeze operator $S(\pu)$ is unitary, we have
$$\langle\eta_\qu^\pu|\eta_\qu^\pu\rangle=\langle S(\pu)\eta_\qu|S(\pu)\eta_\qu\rangle=\langle\eta_\qu|\eta_\qu\rangle=1.$$
That is, the squeezed states are normalized. The dual vector of $|S(\pu)\eta_\qu\rangle$ is $\langle\eta_\qu S(\pu)^\dagger|$. Therefore, from the resolution of the identity of the canonical CS,
$$\int_\quat |\eta_\qu\rangle\langle\eta_\qu| d\zeta(r,\theta,\phi,\psi)=I_{\HI^B_r}$$
we get
$$\int_\quat |S(\pu)\eta_\qu\rangle\langle\eta_\qu S(\pu)^\dagger| d\zeta(r,\theta,\phi,\psi)=S(\pu)I_{\HI^B_r}S(\pu)^\dagger=I_{\HI^B_r}.$$
That is the squeezed states satisfy the resolution of the identity,
$$\int_\quat |\eta_\qu^\pu\rangle\langle\eta_\qu^\pu| d\zeta(r,\theta,\phi,\psi)=I_{\HI^B_r}.$$

\begin{remark}
Since the operators $\D(\pu)$ and $S(\qu)$ are unitary operators the states $\D(\pu)S(\qu)\Phi_0$ are normalized but, for the same reason as for the pure squeezed states, these states cannot hold a resolution of the identity in the space $\HIB$. Further, technically, a series expansion for the states $S(\pu)\D(\qu)\Phi_0$ and $\D(\pu)S(\qu)\Phi_0$ can be obtained. However, it is rather complicated and not necessary (even in the complex case).
\end{remark}
In the complex case, combining the results in the Propositions \ref{dis} and \ref{Psqop} (corresponding complex case) one can obtain a relation for the operators $$S(\xi)^\dagger\D(z)^\dagger\as \D(\xi)S(z)\quad\text{and}\quad D(\xi)^\dagger S(z)^\dagger\asd S(z)\D(\xi)$$ or for the operators $$\D(\xi)^\dagger S(z)^\dagger\as S(z)\D(\xi)\quad\text{ and}\quad \D(\xi)^\dagger S(z)^\dagger\asd S(z)\D(\xi)$$ and use them to compute the expectation values and variances of all the required operators. Since quaternions do not commute such relations cannot be obtained for quaternions. For example, if we combine the Propositions \ref{dis} and \ref{Psqop}, when $\pu=|\pu|e^{i\theta\sigma(\hat{n})}$ let $I_\pu=e^{i\theta\sigma(\hat{n})}$,
\begin{eqnarray*}
\D(\qu)^\dagger S(\pu)^{\dagger}\as S(\pu) \D(\qu) &=&
\D(\qu)^{\dagger} \left[(\cosh{|\pu|})\as+I_\pu\sinh{|\pu|}\cdot\asd\right]\D(\qu)\\
&=&\cosh{|\pu|}\D(\qu)^\dagger\as\D(\qu)+\sinh{|\pu|}\D(\qu)^\dagger I_\pu\cdot\asd\D(\qu).
\end{eqnarray*}
Since $\D(\qu)^\dagger I_\pu\cdot\asd\D(\qu)\not=I_\pu\cdot\D(\qu)^\dagger\asd\D(\qu)$, the above expression cannot be computed. In fact, there is no know technique in quaternion analysis to get a closed form for the expression $\D(\qu)^\dagger I_\pu\cdot\asd\D(\qu)$. In this regard, even though we have established normalized squeezed states $S(\pu)\D(\qu)\Phi_0$ with a resolution of the identity, the corresponding expectation values and variances cannot be obtained in a usable form. Since elements in a quaternion slice commute, if we consider squeezed states in a quaternion slice then the computations can carry forward. From the slice-wise analysis we can get to the whole set of quaternions $\quat$ through direct integrals. For such an analysis with quaternionic canonical coherent states we refer to \cite{MT1}.
%%%%%%%%%%%%%%%%%%%%%%%%%%%%%%%%%%%%%%%%%%%%%%%%%%%%%%%%%%%%%%%%%%%%%%%%%%%%%%%%%%%%%%%%%%%%%%%%%
\section{Squeezed states on a quaternion slice}
Let $\C_I$ be a quaternion slice. Since elements in $\C_I$ commute we can obtain the following relations for squeezed coherent states and two photon coherent states, and obtain the related expectation values. The states $\D(\qu)S(\pu)\Phi_0$ are called the two photon coherent states \cite{Yuen, Rodney}. On the other hand the states $\D(\qu)S(\pu)\Phi_0$ are called the squeezed coherent states \cite{Rodney} pp. 207. We shall demonstrate it briefly in this section.
%%%%%%%%%%%%%%%%%%%%%%%%%%%%%%%%%%%%
\subsection{Two photon coherent states}
Let $\pu,\qu\in\C_I$, then we can write
\begin{eqnarray*}
\pu&=&|\pu|e^{I\theta_{\pu}}=|\pu|I_{\pu}=|\pu|(\cos\theta_\pu+I\sin\theta_\pu)\quad\text{and}\\
\qu&=&|\qu|e^{I\theta_{\qu}}=|\qu|I_{\qu}=|\qu|(\cos\theta_\qu+I\sin\theta_\qu).
\end{eqnarray*}
With these notations we obtain the following.
\begin{proposition} The operators $S(\pu)$ and $\D(\qu)$ satisfies the following relations.
\begin{eqnarray*}
\D(\qu)^\dagger S(\pu)^\dagger\as S(\pu)\D(\qu)&=&\cosh{|\pu|}\as\mathbb{I}_2+I_\pu\sinh{|\pu|}\cdot\asd
+\cosh{|\pu|}\qu\I_2+I_\pu\sinh{|\pu|}\oqu\\
\D(\qu)^\dagger S(\pu)^\dagger\asd S(\pu)\D(\qu)&=&\cosh{|\pu|}\asd\mathbb{I}_2+\overline{I}_\pu\sinh{|\pu|}\cdot\as
+\cosh{|\pu|}\oqu\I_2+\overline{I}_\pu\sinh{|\pu|}\qu,\\
\D(\qu)^{\dagger}S(\pu)^{\dagger}N S(\pu)\D(\qu)&=&\cosh^2{|\pu|}(N+\qu\cdot\asd+\oqu\cdot\as+|\qu|^2)\\
&+&\frac{1}{2}\overline{I}_\pu\sinh{(2|\pu|)}\cdot(\as^2+2\qu\cdot\as+\qu^2)\\
&+&\frac{1}{2}I_\pu\sinh{(2|\pu|)}\cdot((\asd)^2+2\oqu\cdot\asd+\oqu^2)\\
&+&\sinh^2{|\pu|}(\as\asd+\oqu\cdot\as+\qu\cdot\asd+|\qu|^2).
\end{eqnarray*}
\end{proposition}
\begin{proof} Proof is straight forward from the results of the Propositions \ref{Psqop} and \ref{dis}.
\end{proof}
For a normalized squeezed state and an operator $F$ we denote the expectation value as $\langle F\rangle_{\pu\qu}=\langle\eta_\qu^\pu|F|\eta_\qu^\pu\rangle$. The following expectation values can be calculated.
\begin{eqnarray*}
\langle \as\rangle_{\pu\qu}&=&\cosh{|\pu|}\qu+I_\pu\sinh{|\pu|}\oqu\\
\langle \asd\rangle_{\pu\qu}&=&\cosh{|\pu|}\oqu+\overline{I}_\pu\sinh{|\pu|}\qu\\
\langle X\rangle_{\pu\qu}&=&|\qu|\left[\cosh{|\pu|}\cos{\theta_\qu}+\sinh{|\pu|}\cos(\theta_\pu-\theta_\qu)\right]\\
\langle Y\rangle_{\pu\qu}&=&|\qu|\left[\cosh{|\pu|}\sin{\theta_\qu}+\sinh{|\pu|}\sin(\theta_\pu-\theta_\qu)\right]\\
\langle \as\asd\rangle_{\pu\qu}&=&\cosh^2{|\pu|}+\cosh(2|\pu|)|\qu|^2
+|\qu|^2\sinh(2|\pu|)\cos(2\theta_\qu-\theta_\pu)\\
\langle \asd\as\rangle_{\pu\qu}&=&\sinh^2{|\pu|}+\cosh(2|\pu|)|\qu|^2
+|\qu|^2\sinh(2|\pu|)\cos(2\theta_\qu-\theta_\pu)\\
\langle \as^2\rangle_{\pu\qu}&=&\frac{1}{2}I_\pu\sinh{(2|\pu|)}(1+2|\qu|^2)+\cosh^2|\pu|\qu^2
+I_\pu^2\sinh^2|\pu|\oqu^2\\
\langle (\asd)^2\rangle_{\pu\qu}&=&\frac{1}{2}\overline{I}_\pu\sinh{(2|\pu|)}(1+2|\qu|^2)+\cosh^2|\pu|\oqu^2
+\overline{I}_\pu^2\sinh^2|\pu|\qu^2.\\
\end{eqnarray*}
Using the above expectations we can readily obtain the following.
\begin{eqnarray*}
& &\langle X^2\rangle=\frac{1}{2}\left[\left(\cosh(2|\pu|)+2|\qu|^2\cosh(2|\pu|)
+2|\qu|^2\sinh(2|\pu|)\cos(2\theta_\qu-\theta_\pu)\right)\right]\\
& + &\frac{1}{2}\left[\left(\cos\theta_\pu\sinh(2|\pu|)(1+2|\qu|^2)+2|\qu|^2\cosh^2|\pu|\cos(2\theta_\qu)
+2|\qu|^2\sinh^2{|\pu|}\cos(2\theta_\pu-2\theta_\qu)\right)\right]
\end{eqnarray*}
and
\begin{eqnarray*}
& &\langle Y^2\rangle=\frac{1}{2}\left[\left(\cosh(2|\pu|)+2|\qu|^2\cosh(2|\pu|)
+2|\qu|^2\sinh(2|\pu|)\cos(2\theta_\qu-\theta_\pu)\right)\right]\\
& - &\frac{1}{2}\left[\left(\cos\theta_\pu\sinh(2|\pu|)(1+2|\qu|^2)+2|\qu|^2\cosh^2|\pu|\cos(2\theta_\qu)
+2|\qu|^2\sinh^2{|\pu|}\cos(2\theta_\pu-2\theta_\qu)\right)\right].
\end{eqnarray*}
Using these expectation values the variances of $X$ and $Y$ can be obtained.
%%%%%%%%%%%%%%%%%%%%%%%%%%%%%
\subsection{Squeezed coherent states}
The squeezed coherent states are defined as $\eta_\pu^{\qu}=\D(\qu)S(\pu)\Phi_0$ \cite{Rodney, Gaz}. We briefly provide some formulas for these states. Once again we are in a quaternion slice $\C_I$ and $\pu$ and $\qu$ are as in the previous section.
\begin{proposition}
\begin{eqnarray*}
S^{\dagger}(\pu)\D(\qu)^{\dagger}\as\D(\qu)S(\pu)&=&\cosh{|\pu|}~\as+I_\pu\sinh{\pu}~\asd+\qu\\
S^{\dagger}(\pu)\D(\qu)^{\dagger}\asd\D(\qu)S(\pu)&=&\cosh{|\pu|}~\asd+\overline{I}_\pu\sinh{\pu}~\as+\oqu\\
S^{\dagger}(\pu)\D(\qu)^{\dagger}\asd\as\D(\qu)S(\pu)&=&\cosh^2{|\pu|}~\asd\as+\frac{1}{2}I_\pu\sinh(2|\pu|)~(\asd)^2
+\qu\cosh{|\pu|}~\asd\\
&+&\frac{1}{2}\overline{I}_\pu\sinh(2|\pu|)~\as^2+\sinh^2{|\pu|}~\as\asd+\overline{I}_\pu\qu\sinh{|\pu|}~\as\\
&+&\oqu\cosh{|\pu|}~\as+I_\pu\oqu\sinh{|\pu|}~\asd+|\qu|^2.
\end{eqnarray*}
\end{proposition}
\begin{proof}
Proof is straight forward from Propositions \ref{Psqop} and \ref{dis}.
\end{proof}
With the aid of the above proposition we can easily calculate the following.
\begin{eqnarray*}
\langle\as\rangle_{\qu\pu}&=&\qu\\
\langle\asd\rangle_{\qu\pu}&=&\oqu\\
\langle N\rangle_{\qu\pu}&=&\sinh^2{|\pu|}+|\qu|^2\\
\langle X\rangle_{\qu\pu}&=&|\qu|\cos{\theta_\qu}\\
\langle Y\rangle_{\qu\pu}&=&|\qu|\sin{\theta_\qu}\\
\langle \as^2\rangle_{\qu\pu}&=&\frac{1}{2}I_{\pu}\sinh(2|\pu|)+\qu^2\\
\langle (\asd)^2\rangle_{\qu\pu}&=&\frac{1}{2}\overline{I}_{\pu}\sinh(2|\pu|)+\oqu^2\\
\langle \as\asd\rangle_{\qu\pu}&=&\cosh^2{|\pu|}+|\qu|^2.\\
\end{eqnarray*}
Hence, using the above, we can calculate the following.
\begin{eqnarray*}
\langle X^2\rangle_{\qu\pu}&=&\frac{1}{4}\left\{\cosh^2(2|\pu|)+2|\qu|^2
+\sinh(2|\pu|)\cos{\theta_\pu}+2|\qu|^2\cos(2\theta_{\qu})\right\},\\
\langle Y^2\rangle_{\qu\pu}&=&\frac{1}{4}\left\{\cosh^2(2|\pu|)+2|\qu|^2
-\sinh(2|\pu|)\cos{\theta_\pu}-2|\qu|^2\cos(2\theta_{\qu})\right\}
\end{eqnarray*}
and
$$\langle N^2\rangle_{\qu\pu}=\frac{1}{2}\sinh^2(2|\pu|)+\sinh^4|\pu|+2|\qu|^2\sinh^2|\pu|+|\qu|^2\cosh(2|\pu|)
+|\qu|^2\sinh(2|\pu|)+|\qu|^4.$$
Further
\begin{eqnarray*}
\langle \Delta N\rangle_{\qu\pu}^2&=&\frac{1}{2}\sinh^2(2|\pu|)+|\qu|^2\cosh(2|\pu|)
+|\qu|^2\sinh(2|\pu|)\\
&=&\frac{1}{2}\sinh^2(2|\pu|)+|\qu|^2e^{2|\pu|}.
\end{eqnarray*}
For a normalized state $\eta_{\pu}^{\qu}$, in terms of the quadrature operator $X$, the signal-to-noise ratio and the Mandel parameter are, respectively, defined as \cite{Gaz}
$$SNR=\frac{\langle X\rangle_{\qu\pu}^2}{\langle \Delta X\rangle_{\qu\pu}^2}\quad \text{and}\quad
Q_M=\frac{\langle\Delta N\rangle_{\qu\pu}}{\langle N\rangle_{\qu\pu}}-1.$$
Using the above expectation values one can easily obtain these quantities.
%%%%%%%%%%%%%%%%%%%%%%%%%%%%%%%%%%%%%%%%%%%%%%%%%%%%%%%%%%%%%%%%%%%%%%%%%%%%%%%%%%%%%%%%
\begin{proposition}\label{disen}
The operator $S(\pu)$ satisfies the disentanglement formula
\[
S(\pu)= e^{\pu\cdot K_+ - \bar{\pu}\cdot K_-}= e^{\qu\cdot K_+} e^{-2\log(\cosh(2r)) K_0} e^{-\bar{\qu}\cdot K_-}
\]
where $\pu=r\sigma_0e^{i\theta\sigma(\hat n)}=re^{i\theta\sigma(\hat n)}$, $\qu=\tanh(r) e^{i\theta\sigma(\hat n)}$.
\end{proposition}
\begin{proof}
To show the statement we look for $\alpha,\beta,\gamma\in\mathbb H$ such that $\alpha,\beta,\gamma$ mutually commuting, namely they belong to the same slice, such that the formula
\begin{equation}\label{disent}
e^{\pu\cdot K_+ - \bar{\pu}\cdot K_-}= e^{\alpha\cdot K_+} e^{\beta\cdot K_0} e^{\gamma\cdot K_-}
\end{equation}
holds.
We set $A=\pu\cdot K_+ - \bar{\pu}\cdot K_-$ and using the Baker-Campbell-Hausdorff formula we compute
$e^{A}K_0e^{-A}$.
We have that
\[
\begin{split}
&[A, K_0]= -(\pu\cdot K_+ + \bar{\pu}\cdot K_-),\\
&[A,[A, K_0]]=2^2 |\pu|^2K_0\\
&[A,[A,[A, K_0]]]= - 4 |\pu|^2 (\pu\cdot K_+ + \bar{\pu}\cdot K_-)\\
&[A,[A,[A,[A, K_0]]]]= 2^4 |\pu|^4 K_0\\
&[A,[A,[A,[A,[A, K_0]]]]]= -2^4 |\pu|^4 K_0(\pu\cdot K_+ + \bar{\pu}\cdot K_-)\\
& ............
\end{split}
\]
from which we deduce
\begin{equation}\label{disK0}
e^{A}K_0e^{-A}= -\sinh(2r)( e^{i\theta\sigma(\hat n)}\cdot K_+ +e^{-i\theta\sigma(\hat n)}\cdot K_-)+\cosh(2r)K_0.
\end{equation}
We note that the computations mimic the analogous computations in the classical case, since $K_+,K_-,K_0$ belong to $su(1,1)$, the quaternionic variables behaves like a variable commuting with the operators with respect to the left multiplication and the various quaternionic variables are assumed to be mutually commuting. Thus, reasoning as in the classical case, one obtains
\begin{equation}\label{disK-}
e^{A}K_-e^{-A}= \sinh^2(r) e^{i2\theta\sigma(\hat n)}\cdot K_+ +\cosh^2 (r) K_- -\sinh(2r)e^{i\theta\sigma(\hat n)}\cdot K_0.
\end{equation}
Let us denote by $B$ the operator on the right hand side of formula (\ref{disent}) and let us compute $BK_0B^{-1}$ and $BK_-B^{-1}$. To compute the first one we start first by computing $e^{\alpha\cdot K_+}K_0e^{-\alpha\cdot K_+}$. A standard computation shows that
\[
e^{\gamma\cdot K_-}K_0e^{-\gamma\cdot K_-}= \gamma\cdot K_-.
\]
Then one computes $e^{\beta\cdot K_0}(\gamma \cdot K_-)e^{-\beta\cdot K_0}$ and finally $e^{\alpha\cdot K_+}$. The result is
\begin{equation}\label{disK01}
BK_0B^{-1}= (1-2\alpha\gamma e^{-\beta})\cdot K_0+\gamma e^{-\beta}\cdot K_- -\alpha (1-e^{-\beta}\alpha\gamma)\cdot K_+.
\end{equation}
Reasoning in a similar way, and basically using the same computations as in the classical case, we obtain
\begin{equation}\label{disK111}
BK_-B^{-1}= - 2\alpha e^{-\beta}\cdot K_0+e^{-\beta}\cdot K_- + e^{-\beta}\alpha^2\cdot K_+
\end{equation}
By comparing the coefficients obtained in (\ref{disK0}), (\ref{disK-}) and in (\ref{disK01}), (\ref{disK111}) one obtains
 $\alpha=e^{i\theta \sigma(\hat n)}\tanh(r)$, $\beta=-2\log(\cosh(r))$, $\gamma= - e^{-i\theta \sigma(\hat n)}\tanh(r)$ and the statement follows.
\end{proof}
%%%%%%%%%%%%%%%%%%%%%%%%%%%%%%%%%%%%%%%%%%%%%%%%%%%%%%%%%
We will need the above disentanglement formula to realize the connection between squeezed states and Hermite polynomilals in a quaternion slice.
The quaternionic Hermite polynomials are given by 
\begin{equation}\label{qua_her}
H_n(\qu)=n!\sum_{m=0}^{[n/2]}\frac{(-1)^m(2\qu)^{n-2m}}{m!(n-2m)!},\mbox{~~for all~~}\qu\in\quat.
\end{equation}
The following Proposition connects the squeezed states to Hermite polynomilals in a quaternion slice that the squeezed basis vectors are essentially the Hermite polynomials times an exponential function.
\begin{proposition}
For any $\pu\in\mathbb{C}_I$ with $\pu\neq0$, \begin{equation}\label{r1}
\mathfrak{r}=\frac{\tanh(|\pu|)}{|\pu|}\pu,
\end{equation}
then the squeezed basis vector, in the Bargmann analytic (in $\mathbb{C}_I$) representation are given in terms of the complex Hermite polynomials by the expression,
\begin{equation}\label{SqSta_Her}
\Phi_n^{\pu}(\qu)=(S(\pu)\Phi_n)(\qu)=\frac{1}{\sqrt{n!}}(1-|\mathfrak{r}|^2)^{\frac{1}{4}}\left[\frac{\overline{\mathfrak{r}}}{2}\right]^{\frac{n}{2}}e^{\frac{\mathfrak{r}}{2}\qu^2}H_n\left( \left[ \frac{1}{2}(1-|\mathfrak{r}|^2)\overline{\mathfrak{r}}^{-1}\right]^{\frac{1}{2}}\qu \right).  
\end{equation}
\end{proposition}
\begin{proof}
We have , from (\ref{r1}), $\log(1-|\mathfrak{r}|^2)=-2\log\cosh|\mathfrak{r}|$. So the above Proposition (\ref{disen}) enables us to write the squeeze state $S(\pu)$ as
\begin{equation}\label{sqst}
S(\pu)=e^{\frac{\mathfrak{r}}{2}\qu^2}\,e^{\frac{1}{2}\log(1-|\mathfrak{r}|^2)(\qu\partial_s+\frac{1}{2}I_{\HI^B_r})}\,e^{-\frac{\overline{\mathfrak{r}}}{2}\partial_s^2}.
\end{equation}
Now the basis vector $\Phi_n(\qu)=\dfrac{\qu^n}{\sqrt{n!}}$. Further left slice regular derivative of a regular function is regular, and for $\{\mathfrak{a}_m\}\subseteq\quat$, we have, for a right regular power series (see, for example {\cite{Thi1}}), 
\begin{equation}
\partial_s\left(\sum_{m=0}^{\infty}\mathfrak{a}_m\qu^m\right)=\sum_{m=0}^{\infty}m\mathfrak{a}_m\qu^{m-1}. 
\end{equation}
Thus, by doing a right regular power series expansion, we easily obtain, 
\begin{equation}
e^{-\frac{\overline{\mathfrak{r}}}{2}\partial_s^2}\Phi_n(\qu)=e^{-\frac{\overline{\mathfrak{r}}}{2}\partial_s^2}\left(\dfrac{\qu^n}{\sqrt{n!}} \right)= n!\sum_{m=0}^{[n/2]}\frac{(-1)^m\left( \frac{\overline{\mathfrak{r}}}{2}\right)^m (\qu)^{n-2m}}{m!(n-2m)!}. 
\end{equation}
One can note that, for any integer $k$, since $\qu\partial_s\qu^k=k\qu^k$, 
\begin{equation}\label{r2}
e^{\frac{1}{2}\log(1-|\mathfrak{r}|^2)(\qu\partial_s+\frac{1}{2}I_{\HI^B_r})}\qu^k=(\sqrt{1-|\mathfrak{r}|^2})^{k+\frac{1}{2}}\qu^k.
\end{equation}
Combining (\ref{sqst})-(\ref{r2}), and noting (\ref{qua_her}), the result (\ref{SqSta_Her}) follows.
\end{proof}
\section{Conclusion}
 Using the left multiplication on a right quaternionic Hilbert space we have defined unitary squeeze operator. Pure squeezed states have been obtained, with all the desired properties, analogous to their complex counterpart. Even though we have defined squeezed states with the aid of displacement operator and the squeeze operator the noncommutativity of quaternions prevented us in getting desired expectation values and variances. Even though it is a technical issue, there is no known technique to overcome this difficulty. In this regard, the only way out of this difficulty is to consider quaternionic slice-wise approach. We have defined squeezed states on quaternion slices and computed the expectation values of the quadrature operators. We have also proved a quaternionic disentanglement formula.
 
 In the application point of view squeezed states have several applications, particularly in coding and transmission of information through optical devices. These aspects are well explained for example in \cite{Ali, Gaz, Yuen} and the many references therein. Since we have used the matrix representation of quaternions, the squeezed states obtained in this note appear as matrix states. Further these states involve all four variables of quaternions. These features may give advantage in applications.

%\end{multicols}
%%%%%%%%%%%%%%%%%%%%%%%%%%%%%%%%%%%%%%%%%%%%%%%%%%%%%%%%%%%%%%%%%%%%%%%%%%%%%%%%%%
%
%%%%%%%%%%%%%%%%%%%%%%%%%%%%%%%%%%%%%%%%%%%%%%%%%%%%%%%%%%%%%%%%%%%%%%%%%%%%%%%%%%
\end{document}